\documentclass[a4paper,12pt]{article}

\newcommand{\longversion}[1]{#1}
\newcommand{\shortversion}[1]{}

\usepackage{charter}
\usepackage{amssymb,amsmath,amsthm}
\usepackage[usenames,dvipsnames]{color}
\usepackage[colorlinks=true,linkcolor=blue,citecolor=Mahogany]{hyperref}
\usepackage{cleveref}
\usepackage{makecell}
\usepackage{multirow}
\usepackage{mathtools}
\usepackage{eulervm}

\usepackage{tikz}
\usetikzlibrary{shapes,snakes,positioning,arrows.meta,decorations.text}
\usetikzlibrary{calc}
\usetikzlibrary{tikzmark,positioning,arrows}

\DeclareMathOperator*{\argmax}{arg\,max}

\makeatletter
\newcommand\blfootnote[1]{%
	\begingroup
	\renewcommand\thefootnote{}\footnote{#1}%
	\addtocounter{footnote}{-1}%
	\endgroup
}
\makeatother

\longversion{\usepackage{fullpage}}

\usepackage{url}

\usepackage{xspace}

\newcommand{\BigO}{\ensuremath{\mathcal{O}}\xspace}

\newcommand{\el}{\ensuremath{\ell}\xspace}

\newcommand{\suc}{\ensuremath{\succ}\xspace}

\usepackage{MnSymbol,wasysym}

\newcommand{\PF}{\ensuremath{\mathfrak P}\xspace}

\renewcommand{\AA}{\ensuremath{\mathcal A}\xspace}
\newcommand{\BB}{\ensuremath{\mathcal B}\xspace}
\newcommand{\CC}{\ensuremath{\mathcal C}\xspace}

\newcommand{\EE}{\ensuremath{\mathcal E}\xspace}

\newcommand{\HH}{\ensuremath{\mathcal H}\xspace}
\newcommand{\II}{\ensuremath{\mathcal I}\xspace}

\newcommand{\OO}{\ensuremath{\mathcal O}\xspace}
\newcommand{\PP}{\ensuremath{\mathcal P}\xspace}
\newcommand{\QQ}{\ensuremath{\mathcal Q}\xspace}

\renewcommand{\SS}{\ensuremath{\mathcal S}\xspace}
\newcommand{\TT}{\ensuremath{\mathcal T}\xspace}
\newcommand{\UU}{\ensuremath{\mathcal U}\xspace}
\newcommand{\VV}{\ensuremath{\mathcal V}\xspace}

\newcommand{\XX}{\ensuremath{\mathcal X}\xspace}

\newcommand{\fff}{\ensuremath{\mathfrak f}\xspace}

\newcommand{\EB}{\ensuremath{\mathbb E}\xspace}

\usepackage{nicefrac}

\newcommand{\nfrac}{\nicefrac}

\newcommand{\eps}{\varepsilon}
\renewcommand{\epsilon}{\eps}

\newcommand{\ignore}[1]{}

\newcommand{\pr}{\ensuremath{\prime}}
\newcommand{\prr}{\ensuremath{{\prime\prime}}}

\renewcommand{\leq}{\leqslant}
\renewcommand{\geq}{\geqslant}
\renewcommand{\ge}{\geqslant}
\renewcommand{\le}{\leqslant}

\usepackage{cleveref}

\newtheorem{observation}{\bf Observation}
\newtheorem{theorem}{\bf Theorem}
\newtheorem{lemma}{\bf Lemma}

\crefname{property}{Property}{Properties}
\crefname{example}{Example}{Examples}
\crefname{theorem}{Theorem}{Theorems}
\crefname{observation}{Observation}{Observations}
\crefname{lemma}{Lemma}{Lemmas}
\crefname{corollary}{Corollary}{Corollaries}
\crefname{proposition}{Proposition}{Propositions}
\crefname{definition}{Definition}{Definitions}
\crefname{claim}{Claim}{Claims}
\crefname{reductionrule}{Reduction rule}{Reduction rules}
\crefname{ineq}{inequality}{Inequalities}

\numberwithin{property}{section}
\numberwithin{example}{section}
\numberwithin{proposition}{section}
\numberwithin{claim}{section}
\numberwithin{theorem}{section}
\numberwithin{corollary}{section}
\numberwithin{observation}{section}
\numberwithin{lemma}{section}
\numberwithin{definition}{section}
\numberwithin{reductionrule}{section}

\begin{document}

\title{Query Complexity of Tournament Solutions}

\author{Arnab Maiti and Palash Dey\\Indian Institute of Technology, Kharagpur\\
\texttt{\{arnabmaiti,palash.dey\}@cse.iitkgp.ac.in}}

\maketitle

\begin{abstract}
 A directed graph where there is exactly one edge between every pair of vertices is called a {\em tournament}. Finding the ``best'' set of vertices of a tournament is a well-studied problem in social choice theory. A {\em tournament solution} takes a tournament as input and outputs a subset of vertices of the input tournament. However, in many applications, for example, choosing the  best set of drugs from a given set of drugs, the edges of the tournament are given only implicitly and knowing the orientation of an edge is costly. In such scenarios, we would like to know the best set of vertices (according to some tournament solution) by ``querying'' as few edges as possible. We, in this paper, precisely study this problem for commonly used tournament solutions: given an oracle access to the edges of a tournament \TT, find $\fff(\TT)$ by querying as few edges as possible, for a tournament solution \fff. We first study some popular tournament solutions and show that any deterministic algorithm for finding the Copeland set, the Slater set, the Markov set, the bipartisan set, the uncovered set, the Banks set, and the top cycle must query $\Omega(n^2)$ edges in the worst case. We also show similar lower bounds on the expected query complexity of these tournament solutions by any randomized algorithm. On the positive side, we are able to circumvent our strong query complexity lower bound results by proving that, if the size of the top cycle of the input tournament is at most $k$, then we can find all the tournament solutions mentioned above by querying $\BigO(nk +\nfrac{n\log (\nfrac{n}{k}) }{\log(\nfrac{k}{k-1})})$ edges only.
\end{abstract}


\newpage
\blfootnote{A preliminary version of this work appeared in the proceedings of AAAI Conference on Artificial Intelligence 2017~\cite{DBLP:conf/aaai/Dey17}. Our lower-bound results on the expected query complexity of any randomized algorithm for these tournament solutions in \Cref{thm:random_many_st_lb,thm:random_bps_lb,thm:random_ucs_bs_lb,thm:random_topcycle_lb} are new which were not present in the preliminary version.}

\section{Introduction}


Many scenarios in multiagent systems can often be modeled and analyzed using tournaments~\cite{moulin86,brandt2014minimal}. An important example of such scenarios is {\em voting} where we have a set of alternatives and a set of votes which are linear orders over the set of alternatives. An important tournament to consider in this context is defined by the {\em majority relation} induced by the set of votes. In the majority relation, an alternative $x$ is preferred over another alternative $y$ if $x$ is preferred over $y$ in a majority of the votes. Indeed, many important voting rules, like Copeland for example, are defined using the tournament induced by the majority relation of the input set of votes. Other than voting, tournaments have found many applications in multi-criteria decision analysis~\cite{arrow1986social,bouyssou2006evaluation}, zero-sum games~\cite{fisher1995tournament,duggan1996dutta}, coalitional games~\cite{brandt2010characterization}, argumentation theory~\cite{dung1995acceptability,dunne2007computational}, biology~\cite{charon2007survey}, etc.

Formally, a {\em tournament} is defined as a set of alternatives along with an {\em irreflexive, antisymmetric, and total relation}, called {\em dominance relation}, on the set of alternatives. An equivalent and often more convenient view of a tournament is as a {\em directed} graph on the alternatives where, between every pair of vertices (which corresponds to the alternatives), there is exactly one edge. A {\em tournament solution} takes a tournament as input and outputs a subset of the vertices. Tournaments and tournament solutions are important mathematical tools in any general situation where we have to make a choice from a set of alternatives solely based on pairwise comparisons. We refer the reader to \cite{laslier97,brandt2016tournament,hudry2009survey,suksompong2021tournaments} for more details on tournament solution.\\
\longversion{\subsection{Motivation}}
\shortversion{{\bf Motivation.~~}}
We often have situations where the input tournament is given ``implicitly'' -- the vertices of the tournament are given explicitly and we have to query for an edge to know its orientation. Moreover, knowing the orientation of an edge of the tournament can often be costly. For example, we can think of an application where we have a set of drugs for a particular disease and we want to know the ``best'' (according to some tournament solution) set of drugs. A natural dominance relation in this context would be to define a drug $\gamma$ to dominate another drug $\eta$ if the probability that the drug $\gamma$ cures the disease is more than the corresponding probability for the drug $\eta$. Since these probabilities are often not known a priori, estimating them often requires extensive lab experiments as well as clinical trials. Hence, we would like to make as few queries as possible to know the best set of drugs. More generally, we can think of any {\em tournament based voting rules} like Copeland in an election scenario. A tournament based voting rule chooses winners solely based on the tournament induced by the pairwise majority relation between the alternatives. However, in many applications of voting in multiagent systems, recommender systems~\cite{PennockHG00} for example, the number of voters is huge and consequently, learning the majority relation is costly. Hence, we would like to learn the set of most popular items (according to the tournament solution under consideration) with the smallest number of queries possible. Motivated by these applications, we study, for a tournament solution \fff, the problem of finding $\fff(\TT)$ of an input tournament \TT by querying the smallest number of edges possible. A {\em query} for an edge, in our model, reveals the orientation of the edge in the input tournament.

Finding the query complexity of various graph properties has drawn significant attention in the literature. In the most general setting, the input is a directed graph on $n$ vertices and one has to find whether the input graph satisfies some property we are concerned with, by asking a minimum number of queries. A query is a question of the form: ``Is there an edge from a vertex $x$ to another vertex $y$?'' The {\em query complexity} of a property is the worst case number of queries that must be made to know whether the input graph has that property. A graph property, in this context, is called {\em evasive} if its query complexity is $n(n-1)$, that is, one has to query all possible edges of the input digraph to test the property in the worst case. Karp conjectured that every monotone and nontrivial graph property is evasive~\cite{rosenberg1973time}. A graph property is {\em monotone} if it continues to hold even after adding more edges and {\em nontrivial} if it holds for some but not all graphs. A substantial amount of research effort has provided increasingly better query complexity lower bounds for monotone and nontrivial properties, although Karp's conjecture remained open~\cite{rosenberg1973time,rivest1976recognizing,kahn1984topological,king1988lower,chakrabarti2001evasiveness,korneffel2010asymptotic,Kulkarni}. In the case of tournaments (where there is exactly one edge between every pair of vertices), Balasubramanian et al.~\cite{DBLP:journals/jal/BalasubramanianRS97} showed (rediscovered by Procaccia~\cite{procaccia2008note}) that deciding if there exists a Condorcet winner --- a vertex with $n-1$ outgoing edges -- can be computed with $2n-\lfloor \log n \rfloor -2$ queries and this query complexity upper bound is tight in the worst case. This further motivates us to study the query complexity of other commonly used tournament solutions.\\
\longversion{\subsection{Our Contribution}}
\shortversion{{\bf Our Contribution.~~}}
In this paper, we prove tight bounds on the query complexity of commonly used tournament solutions. Our specific contributions in this paper are as follows.

\begin{itemize}
 \item We show that the query complexity of the problem of finding the set of Condorcet non-losers is $2n-\lfloor \log n \rfloor -2$ [\Cref{obs:cond_win_lose}].

 \item We show that any deterministic algorithm for finding the Copeland set, the Slater set, and the Markov set in a tournament has query complexity ${n\choose 2}$ [\Cref{thm:many_st_lb}]. We remark that Goyal et al.~\cite{goyalCaldem17} also discovered this result independently (and in parallel) around the same time. We also show that any randomized algorithm for finding the Copeland set, the Slater set, and the Markov set of tournaments has an expected query complexity of ${n\choose 2}$ [\Cref{thm:random_many_st_lb}].

 \item We prove that any deterministic algorithm for finding the bipartisan set [\Cref{thm:bps_lb}], the uncovered set [\Cref{thm:ucs_lb}], the Banks set [\Cref{thm:banks_lb}], and the top cycle [\Cref{thm:topcycle_lb}] has query complexity $\Omega(n^2).$ We then extend all these lower bounds to randomized algorithms without deteriorating the lower bounds~[\Cref{thm:random_bps_lb,thm:random_ucs_bs_lb,thm:random_topcycle_lb}].

 \item We complement our strong query complexity lower bounds above by showing that, if the tournament \TT has a top cycle of size at most $k$, then the Copeland set, the Slater set, the Markov set, the bipartisan set, the uncovered set, the Banks set, and the top cycle of a tournament \TT can be found using $\BigO(nk + \nfrac{n\log (\nfrac{n}{k}) }{\log(\nfrac{k}{k-1})})$ queries [\Cref{thm:topcycle_ub,thm:others_cycle_ub}].
\end{itemize}
\longversion{\subsection{Related Work}}
\shortversion{{\bf Related Work.~~}}
The work of Balasubramanian et al.~\cite{DBLP:journals/jal/BalasubramanianRS97} (rediscovered by Procaccia~\cite{procaccia2008note}) is the closest predecessor of our work where they show that the query complexity of Condorcet winner in tournaments is $2n-\lfloor \log n \rfloor -2.$ Goyal et al.~\cite{goyalCaldem17} show that all the edges need to be queried to find a vertex with some specific degrees. There have been several other works in the literature which study the query complexity of various problems in tournaments. For example, Noy and Naor use comparison-based sorting algorithms to find feedback sets and Hamiltonian paths in tournaments~\cite{bar1990sorting}. There have been works in computational social choice theory on communication complexity of different voting rules~\cite{conitzer2002vote} and query complexity of preference elicitation in various domains~\cite{Conitzer09,deypeak,deycross,DBLP:journals/corr/DeyM16}. However, the querying model in the above works is completely different from ours and consequently, neither the results nor the techniques involved in these works are directly applicable to our work; a query in the works above asks: ``Does a voter $v$ prefer an alternative $x$ over another alternative $y$?''

Recently, \cite{brill2022margin} worked on Margin of Victory in tournament solutions where the aim is to determine how robust a tournament solution is when the tournament is subject to changes. Kondratev and Mazalov studied manipulating various power indices, for example the Shapley–Shubik index, the Penrose–Banzhaf index, and the nucleolus, via forming coalitions in a tournament setting~\cite{kondratev2020tournament}. Saile and Suksompong studied a number of common tournament solutions, including the top cycle and the uncovered set in a model where the orientation of different edges can be chosen with different probabilities and showed that these tournament solutions are still unlikely to rule out any alternative under this model~\cite{saile2020robust}. Han and Van Deemen proposed a new solution for tournaments called the unsurpassed set which
lies between the uncovered set and the Copeland winner set~\cite{han2019refinement}. Manurangsi and Suksompong showed important characterization results for generalized kings in tournaments~\cite{manurangsi2022generalized}. There exist a large body of research on manipulating tournaments~\cite[and references therein]{manurangsi2023fixing,lisowski2022strategic}. We refer to~\cite{suksompong2021tournaments} for a survey of recent developments in tournament solutions. Last but not the least, there have been quite a number of papers on probabilistic pairwise comparison queries and for this we refer the reader to the survey \cite{busa2014survey}.

%
%
%
%
%
%
%
%
%

\subsection*{Organization} In \Cref{sec:prelim}, we formally define tournaments, tournament solutions; we present our query complexity lower bounds in \Cref{sec:qlb}; we show in \Cref{sec:small_cycle} how the presence of a small top cycle enables us to reduce the query complexity of all the tournament solutions studied in this paper; we finally conclude with future directions in \Cref{sec:con}. A preliminary version of this work appeared in the proceedings of AAAI Conference on Artificial Intelligence 2017~\cite{DBLP:conf/aaai/Dey17}. Our lower-bound results on the expected query complexity of any randomized algorithm for these tournament solutions in \Cref{thm:random_many_st_lb,thm:random_bps_lb,thm:random_ucs_bs_lb,thm:random_topcycle_lb} are new which were not present in the preliminary version.

\section{Preliminaries}\label{sec:prelim}

For a positive integer \el, we denote the set $\{1, 2, \ldots, \el\}$ by $[\el].$ For a finite set \XX and a positive integer $\el$, we denote the set of all subsets of \XX of size \el by $\PF_\el(\XX)$ and the set of all probability distributions on \XX by $\Delta(\XX).$

\subsection{Tournaments and Tournament Solutions} A tournament $\TT = (\VV, \EE)$ is a directed graph on a set \VV of $n$ vertices such that, for any two vertices $u, v\in\VV$, either $(u,v)\in\EE$ or $(v,u)\in\EE$ but not both. We call any subgraph of a tournament a {\em partial tournament}. We call a tournament {\em regular} if the in-degree and out-degree of every vertex are the same. A tournament \TT defines a relation $\suc_\TT$ on the set of vertices \VV: $u\suc_\TT v$ if $(u,v)\in\TT$. Alternatively, any irreflexive, antisymmetric, and total relation \suc on a set \VV defines a tournament $\TT = (\VV, \EE)$ where $(u,v)\in\EE$ if $u\suc v$. When there is no possibility of confusion, we drop \TT from the subscript of \suc. We call the relation $\suc_\TT$ associated with a tournament \TT the {\em dominance relation} of \TT. We say a vertex $u$ {\em dominates or defeats} another vertex $v$ if $u\suc_\TT v.$ Let us define the {\em dominion} $D(v)$ of a vertex $v$ as $D(v)=\{u\in\VV: v\suc u\}$ and $\overline{D(v)}=\{u\in\VV: u\suc v\}$ is called the {\em dominators} of $v$. $|D(v)|$ is also known as the Copeland score of $v$. Given a tournament \TT, let $\UU(\TT)$ be its adjacency matrix. Recall that $\UU(\TT)_{ij}=1$ if there is an edge from $i$ to $j$; otherwise $\UU(\TT)_{ij}=0$. We call the matrix $\PP(\TT) = \UU(\TT) - \UU(\TT)^t$ the {\em skew-adjacency matrix} of \TT, where $\UU(\TT)^t$ is the transpose of $\UU(\TT)$. A vertex $v$ is called the {\em Condorcet winner} of a tournament if the out-degree of $v$ is $n-1$; alternatively, a vertex $v$ is the Condorcet winner if $v\suc u$ for every $u\in\VV\setminus\{v\}$. Given a set \VV, we denote the set of all tournaments over \VV by ${\frak T}(\VV)$. A {\em tournament solution} $\fff:\bigcup_{|\VV|>0} {\frak T}(\VV) \rightarrow 2^\VV\setminus\{\emptyset\}$ is a function that takes a tournament as input and selects a nonempty set of vertices as output. Examples of commonly used tournament solutions are as follows~\cite{brandt2016tournament}.

\begin{itemize}
 \item {\bf Condorcet non-loser:} The Condorcet non-loser set of a tournament is the set of all vertices which has at least one outgoing edge.

 \item {\bf Copeland set:} The Copeland set of a tournament is the set of vertices with maximum out-degree.

 \item {\bf Slater set:} Given a tournament $\TT=(\VV,\suc),$ let us denote the maximal element of \VV according to a strict linear order $\suc^\pr$ on \VV by max($\suc^\pr$). The Slater score of a strict linear order $\suc^\pr$ over \VV with respect to tournament \TT = (\VV, \suc) is $|\suc^\pr \cap \suc|$ where $\suc^\pr \cap \suc = \{(x,y)\in\VV\times\VV: x\suc^\pr y, x\suc y\}$. A strict linear order is a Slater order if it has maximum Slater score. Then the Slater set is defined as SL(\TT) = \{max($\suc^\pr$) : $\suc^\pr$ is a Slater order for \TT\}.

 \item {\bf Markov set:} Given a tournament $\TT=(\VV,\EE),$ we define a Markov chain ${\frak M(\TT)}$ as follows. The states of the Markov chain ${\frak M(\TT)}$ are the vertices of \TT and the transition probabilities are determined by the dominance relation: in every step, stay in the current state $v$ with probability $\nfrac{|D(v)|}{|\VV|-1}$, and move to state $u$ with probability $\nfrac{1}{|\VV|-1}$ for all $u\in\overline{D(v)}$. The Markov set is the set of vertices that have maximum probability in the unique stationary distribution of ${\frak M(\TT)}$. Formally, the transition matrix of the Markov chain is defined as.
 \[\QQ = \frac{1}{|\VV|-1}\left(\UU(\TT)^t+diag(CO)\right)\]
 where $diag(CO)$ is the diagonal matrix of the Copeland scores. The Markov set MA(\TT) of a tournament \TT is then given by $MA(\TT) = \argmax_{a\in\VV}\{p(a) : p\in\Delta(\VV), \QQ p = p\}.$

 \item {\bf Bipartisan set:} The Bipartisan set generalizes the notion of Condorcet winner to lotteries over the vertices of the tournament. Interestingly, for every tournament $\TT=(\VV,\EE),$ there exists a unique {\em maximal lottery}~\cite{brandt2016tournament}. That is, there exists a probability distribution $p\in\Delta(\VV)$ such that, for the skew-adjacency matrix $\PP(\TT) = (g_{ab})_{a,b\in\VV}$ of \TT, $\sum_{a,b\in\VV} p(a)q(b)g_{ab} \ge 0 \text{ for all } q\in \Delta(\VV)$ which can be shown to be equivalent to the following condition.

 \begin{equation}
\sum_{a\in\VV} p(a)g_{ab} \ge 0 \text{ for all } b\in\VV\label{eqn:bps_cond}
 \end{equation}

 Let $p_\TT$ denote the unique maximal lottery of \TT. Then the bipartisan set BP(\TT) of \TT is defined as the support of $p_\TT$. That is,
 \[ BP(\TT) = \{a\in\VV : p_\TT(a)>0 \} \]

 \item {\bf Uncovered set:} Given a tournament \TT=(\VV,\suc), we say a vertex $v\in\VV$ covers another vertex $u\in\VV$ if $D(u) \subseteq D(v)$; this is denoted by $vCu$. We observe that $vCu$ implies $v\suc u$ and is equivalent to $\overline{D(v)} \subseteq \overline{D(u)}$. The uncovered set UC(\TT) of a tournament \TT is given by the set of maximal elements of the covering relation $C$. That is,
 \[UC(\TT) = \{v\in\VV : uCv \text{ for no } u\in\VV\}\]

 \item {\bf Banks set:} A sub-tournament of a tournament $\TT = (\VV, \EE)$ is an induced subgraph of \TT. The Banks set BA(\TT) is the set of maximal elements of all the maximal transitive sub-tournaments of \TT.

 \item {\bf Top cycle:} A non-empty subset of vertices $\BB\subseteq \VV$ is called dominant in a tournament $\TT=(\VV,\suc)$ if $ x\suc y$ for every $x\in\BB$ and $y\in\VV\setminus\BB$. Dominant sets are linearly ordered via set inclusion and the top cycle returns the unique smallest dominant set.
\end{itemize}

A tournament solution is called {\em neutral} if the output does not depend on the names of the input set of vertices. All the tournament solutions discussed above are neutral.

\subsection{Essential States in a Markov Chain} A state $i$ in a finite Markov chain is called {\it essential} if for every state $j$ that is reachable from $i$, the state $i$ is also reachable from $j$. A state is called inessential if it is not essential. A well known fact from probability theory is that, $\pi(i)=0$ for every inessential state $i$, where $\pi$ is a stationary distribution of the Markov chain. Hence, every vertex whose corresponding state in the Markov chain is inessential, does not belong to the Markov set of that tournament. We refer the reader to \cite{bremaud2013markov} for a more detailed discussion on Markov chains.\\

\subsection{Query Model} Given a tournament $\TT = (\VV, \EE)$ on $n$ vertices, a query for a pair of vertices $\{x,y\}\in\PF_2\{\VV\}$ reveals whether $(x,y)\in\EE$ or $(y,x)\in\EE$. The {\em query complexity of an algorithm} is the maximum number of queries the algorithm makes in the worst case. The {\em query complexity of a tournament solution} \fff is the minimum query complexity of any algorithm for computing \fff.

\subsection{Yao's Lemma}

We use Yao's lemma to prove our expected query complexity lower bounds. Intuitively speaking, Yao's lemma states that, to prove a lower bound on the expected cost of any randomized algorithm for some problem, it is enough to prove a lower bound on the average cost of any deterministic algorithm for that problem according to some distribution over inputs.

\begin{lemma}[Yao's Lemma~\cite{DBLP:conf/focs/Yao77}]
	Let $A$ be a randomized algorithm for some problem $\Pi$. For an input $x$, let $T(A,x)$ be the random variable denoting the cost (e.g. running time, number of comparisons made) of $A$ on $x$. Let \XX be the set of all inputs (of length $n$) to $\Pi$, $X$ be a random variable denoting the input chosen from \XX according to some distribution $p$ on \XX, and $\AA_{\Pi}$ be the set of all deterministic algorithms for the problem $\Pi$. Then we have the following.
	\[ \max_{x\in\XX} \EB[T(A,x)] \ge \min_{a\in\AA_{\Pi}} \EB[T(a,X)] \]\label{yao}
\end{lemma}

\section{Query Complexity Lower Bounds of Tournament Solutions}\label{sec:qlb}

We now present our lower bounds. If not mentioned otherwise, by algorithms, we mean deterministic algorithms. We begin with the following observation which gives us the query complexity of the Condorcet non-loser set of tournaments.

\begin{observation}\label{obs:cond_win_lose}
 The query complexity of the Condorcet non-loser set in tournaments is equal to the query complexity of the Condorcet winner in tournaments. Hence, the query complexity of the Condorcet non-loser set in tournaments is $2n-\lfloor \log n \rfloor -2$.
\end{observation}

\begin{proof}
 Given a tournament $\TT=(\VV,\EE)$, let us define another tournament $\overline{\TT} = \{\VV, \overline{\EE}\}$, where $\overline{\EE} = \{(x, y): (y,x)\in\EE\}.$ Now the result follows from the observation that a vertex $v$ is a Condorcet non-loser in \TT if and only if $v$ is not the Condorcet winner in $\overline{\TT}.$ Now the $2n-\lfloor \log n \rfloor -2$ query complexity of the Condorcet non-loser set follows from the $2n-\lfloor \log n \rfloor -2$ query complexity of the Condorcet winner in tournaments~\cite{DBLP:journals/jal/BalasubramanianRS97}.
%
\end{proof}

%
%
%
%

We next consider the query complexity of the Slater set of tournaments. The following result provides a necessary condition for a vertex to belong to the Slater set of a tournament~\cite{laslier97}. We provide a short proof here for easy accessibility. We will use it to prove a query complexity lower bound of the Slater set.\shortversion{ In the interest of space, we skip the proof of \Cref{lem:Slater}.}

\begin{lemma}\label{lem:Slater}
 Suppose the out-degree of a vertex $v\in\VV$ in a tournament $\TT=(\VV,\EE)$ on $n$ vertices is strictly less than $\nfrac{(n-1)}{2}$. Then $v$ does not belong to the Slater set of \TT.
\end{lemma}
\longversion{
\begin{proof}
 Let $\overrightarrow{\VV\setminus\{v\}}$ be any arbitrary strict linear order on the vertices in $\VV\setminus\{v\}$. Let $(v_1\suc v_2 \suc\cdots \suc v_k)$ denote the set of edges $\{(v_i,v_j):1\leq i<j\leq k\}$. Now it is easy to observe that $|(v\suc \overrightarrow{\VV\setminus\{v\}}) \cap \EE| < |(\overrightarrow{\VV\setminus\{v\}}\suc v) \cap \EE|$ as the out-degree of vertex $v$ is strictly less than $\nfrac{(n-1)}{2}$. This implies that there is a strict linear order whose maximal element is not $v$ whose Slater score is greater then any strict linear order with $v$ as the maximal element.
\end{proof}
}
Let us consider the tournament $\TT_{reg} = (\VV, \EE)$ on $n$ vertices $\VV = \{a_i: i\in[n]\}$. We assume $n$ to be an odd integer. In $\TT_{reg},$ vertex $a_i$ defeats $a_{i+j\pmod n}$ for every $i\in[n]$ and $j\in[\nfrac{(n-1)}{2}].$ Hence, every vertex has the same in and out degree in $\TT_{reg}$. We will use the tournament $\TT_{reg}$ crucially in our proofs of the query complexity lower bounds of the Copeland set, the Slater set, and the Markov set. The following result is immediate from the definition of neutral tournament solutions.

\begin{lemma}\label{lem:neutral}
 Given the tournament $\TT_{reg}$ as input, the Copeland set, the Slater set, and the Markov set contain the set of all vertices in $\TT_{reg}$.
\end{lemma}


Suppose there exists an edge from a vertex $u$ to another vertex $v$ in $\TT_{reg}$. Let $\TT_{reg}^{u,v}$ be the tournament which is the same as $\TT_{reg}$ except the edge from $u$ to $v$ is reversed, that is, $\TT_{reg}^{u,v} = \TT_{reg}\cup\{(v,u)\}\setminus\{(u,v)\}$. The following lemma will be used crucially in our proofs of the query complexity lower bounds of the Copeland set, the Slater set, and the Markov set.

\begin{lemma}\label{lem:reg_modified}
 The Copeland set, the Slater set, and the Markov set of $\TT_{reg}^{u,v}$ do not contain $u$.
\end{lemma}

\begin{proof}
\begin{itemize}
 \item {\bf Copeland set:} For the Copeland set, the result follows from the observation that $u$ is no longer a vertex with highest out-degree in $\TT_{reg}^{u,v}$.

 \item {\bf Slater set:} For the Slater set, the result follows immediately from \Cref{lem:Slater} since the out-degree of $u$ is strictly smaller than $\nfrac{(n-1)}{2}$ in $\TT_{reg}^{u,v}$.

 \item {\bf Markov set:} Follows from the fact that every vertex in the Markov set has out-degree at least $\nfrac{(n-1)}{2}$~\cite{DBLP:journals/siamdm/KimSW17}.\qedhere
\end{itemize}
\end{proof}

We now prove that any algorithm for finding the Copeland set, the Slater set, and the Markov set must query every edge in the input tournament in the worst case.

\begin{theorem}\label{thm:many_st_lb}
 Any algorithm for finding the Copeland set, the Slater set, and the Markov set of tournaments on an odd number of vertices has query complexity ${n\choose 2}$.
\end{theorem}

\begin{proof}
 Let us consider the tournament $\TT_{reg}$. We observe that, due to \Cref{lem:neutral}, the Copeland set, the Slater set, and the Markov set of $\TT_{reg}$ is the set of all vertices. Let us now consider the oracle which answers the queries according to $\TT_{reg}$. We claim that any algorithm for finding the Copeland set, the Slater set, and the Markov set of tournaments must query all the ${n\choose 2}$ edges of $\TT_{reg}$. Suppose not, then there exists an edge from $u$ to $v$ in $\TT_{reg}$ for some vertices $u$ and $v$ that the algorithm does not query. Let $\hat{\TT}$ be the partial tournament of $\TT_{reg}$ containing the edges that the algorithm queries. If the output of the algorithm does not contain $u$, then the oracle completes $\hat{\TT}$ to $\TT_{reg}$ and thus the algorithm does not output correctly since the output should contain all the vertices. On the other hand, if the output of the algorithm contains $u$ then the oracle completes $\hat{\TT}$ by directing the edge between $u$ and $v$ from $v$ to $u$ and directing the rest of the edges as in $\TT_{reg}$. Again the algorithm outputs wrongly due to \Cref{lem:reg_modified}.
\end{proof}

In the proof of \Cref{thm:many_st_lb}, we have crucially used the assumption that the number of vertices is an odd integer. For tournaments on an even number of vertices, we immediately obtain a lower bound of ${n-1\choose 2}$ by simply constructing a tournament having one vertex which loses to everyone else. However, it may be possible to compute Copeland or Slater or Markov sets on a tournament on an even number of vertices with less than ${n\choose 2}$ queries. We now extend \Cref{thm:many_st_lb} to randomized algorithms using \Cref{yao}.

\begin{theorem}\label{thm:random_many_st_lb}
 Any randomized algorithm for finding the Copeland set, the Slater set, and the Markov set of tournaments on an odd number of vertices has an expected query complexity of ${n\choose 2}$.
\end{theorem}

\begin{proof}
 We present the proof for the Markov set. The proofs for the Copeland set and the Slater set are the same.

 Fix a constant $\varepsilon$ such that $0<\varepsilon<1$. Let $\II:=\{\TT_{reg}\}\cup\{\TT_{reg}^{u,v}:(u,v)\in \EE[\TT_{reg}]\}$ be a set of input instances. Now let us define a distribution on \II. $\TT_{reg}$ is considered as the input instance with probability $1-\varepsilon$. For all $\TT'\in \II\setminus \{\TT_{reg}\}$, $\TT'$ is considered as the input instance with probability $\frac{\varepsilon}{{n\choose 2}}$. Let \AA be a deterministic algorithm that always computes the Markov set. We consider the oracle which answers the queries according to $\TT_{reg}$. We claim that \AA must query all the ${n\choose 2}$ edges of $\TT_{reg}$. Suppose not, then there exists an edge from $u$ to $v$ in $\TT_{reg}$ for some vertices $u$ and $v$ that the algorithm does not query. Then \AA will output the same solution for both $\TT_{reg}$ and $\TT_{reg}^{u,v}$. Due to   \Cref{lem:neutral} and \Cref{lem:reg_modified}, we can conclude that with non-zero probability, the algorithm \AA will output a wrong solution. This is a contradiction to our assumption that \AA is a deterministic algorithm that always computes the correct solution.  Hence, \AA must query all the ${n\choose 2}$ edges of $\TT_{reg}$. The expected query complexity is at least $(1-\varepsilon){n\choose 2}$. Since this holds for any $\varepsilon$ such that $0<\varepsilon<1$, the expected query complexity is ${n\choose 2}$.
\end{proof}

We now present our query complexity lower bound for the bipartisan set of tournaments. Before embarking on the query complexity lower bound of the bipartisan set, let us prove a few structural results for the bipartisan set which we will use crucially later. The following result for the bipartisan set is well known \cite{laffond93,fisher1995tournament}.

\begin{lemma}\label{lem:bps_uc}
 In a tournament $\TT=(\VV,\suc)$, suppose a vertex $u$ covers another vertex $v$, that is, $u\suc w$ whenever $v\suc w$ for every $w\in\VV$. Then $v$ does not belong to the bipartisan set of \TT.
\end{lemma}

The following result shows that, in the tournaments where every vertex has the same number of incoming edges as the number of outgoing edges, the bipartisan set is the set of all vertices.

\begin{lemma}\label{lem:bps_same}
 Let \TT be a tournament on $n$ vertices where the in-degree and out-degree of every vertex is $\nfrac{(n-1)}{2}$. Then the only maximal lottery of \TT is the uniform distribution over the set of all vertices of \TT and thus the bipartisan set of \TT is the set of all vertices.
\end{lemma}

\begin{proof}
 We observe that the uniform distribution over all the vertices satisfies \Cref{eqn:bps_cond} for the tournament \TT, since the average of the entries in every column of the skew-symmetric matrix $\PP$ of the tournament \TT is $0$. Now the result follows from the fact that the maximal lottery in every tournament is unique~\cite{laffond93,fisher1995tournament}.
\end{proof}

In the next lemma, we formalize the intuition that, if a $(\AA, \VV\setminus\AA)$ cut in a tournament has a majority of its edges from $\VV\setminus\AA$ to \AA and \AA is regular, then the bipartisan set of the tournament must include at least one vertex from $\VV\setminus\AA.$

\begin{lemma}\label{lem:heavy}
 Let $\TT = (\VV_1 \cup \VV_2, \EE)$ be a tournament such that there exist at most $(\nfrac{|\VV_1|\cdot|\VV_2|}{2})-1$ edges from $\VV_1$ to $\VV_2$ and the in-degrees and out-degrees of all the vertices in the sub-tournament $\TT[\VV_1]$ of \TT induced on $\VV_1$ are exactly $\nfrac{(|\VV_1|-1)}{2}.$ Then the bipartisan set of \TT must include at least one vertex from $\VV_2.$
\end{lemma}

\begin{proof}
 Let $p^\star\in\Delta(\VV_1 \cup \VV_2)$ be the maximal lottery of \TT and $\VV = \VV_1 \cup \VV_2$. If possible, let us assume that $p^\star(v)=0$ for every $v\in\VV_2.$ Let $q\in\Delta(\VV_2)$ be the uniform distribution on $\VV_2$ and $\PP=(g_{ab})_{a,b\in\VV}$ the skew-adjacency matrix of \TT. We first claim that $p^\star$ cannot be the uniform distribution on $\VV_1.$ For the sake of arriving to a contradiction, let us assume that $p^\star$ is a uniform distribution on $\VV_1$. Then we have
 $$\sum_{a,b\in\VV} p^\star(a)q(b)g_{ab}=\frac{1}{|\VV_1|}\frac{1}{|\VV_2|}\sum_{a,b\in\VV} g_{ab} < 0$$
 since a strict majority of the edges between $\VV_1$ and $\VV_2$ are from $\VV_2$ to $\VV_1$. However, this contradicts the fact that $p^\star$ is the maximal lottery of \TT. Hence, $p^\star$ is not a uniform distribution on $\VV_1$. We now consider the sub-tournament $\TT[\VV_1]$ of \TT induced on $\VV_1$. Due to \Cref{lem:bps_same}, the only maximal lottery of $\TT[\VV_1]$ is the uniform distribution on $\VV_1.$ Hence, since $p^\star$ is not the uniform distribution over $\VV_1,$ there exists a distribution $q^\pr\in\Delta(\VV_1)$ such that $\sum_{a,b\in\VV_1} p^\star(a)q^\pr(b)g_{ab} < 0.$ However, this contradicts our assumption that $p^\star$ is the maximal lottery of \TT. Hence, the bipartisan set of \TT must include at least one vertex from $\VV_2.$
\end{proof}

We now have all the ingredients to present our query complexity lower bound result for the bipartisan set.

\begin{figure}[htbp]
\begin{center}
\begin{tikzpicture}

  \draw[ultra thick] (0,2) ellipse (.6cm and 1.5cm);
  \draw[ultra thick] (2.8,2) ellipse (.6cm and 1.5cm);

  \node at (0,2) {$\AA$};
  \node at (2.8,2) {$\BB$};

  \draw[<-] (2.6,3.5) -- (.2,3.5);
  \draw[<-] (2.2,3) -- (.5,3);
  \draw[<-] (2.1,2.5) -- (.7,2.5);
  \draw[<-] (2.1,2) -- (.7,2);
  \draw[<-] (2.1,1.5) -- (.7,1.5);
  \draw[->] (.5,1) -- (2.2,1);
  \draw[->] (.2,.5) -- (2.5,.5);

  \draw[ultra thick] (7.8,2) ellipse (.8cm and 1.5cm);
  \draw[ultra thick] (10.8,2) ellipse (.8cm and 1.5cm);

  \node at (7.8,2) {$\AA$};
  \node at (10.8,2) {$\BB$};

  \draw[<-] (8.1,3.5) -- (10.5,3.5);
  \draw[<-] (8.5,3) -- (10,3);
  \draw[<-] (8.7,2.5) -- (9.8,2.5);
  \draw[<-] (8.7,2) -- (9.8,2);
  \draw[<-] (8.7,1.5) -- (9.8,1.5);
  \draw[dashed,->] (8.5,1) -- (10,1);
  \draw[dashed,->] (8.1,.5) -- (10.5,.5);
  \end{tikzpicture}
\end{center}
\caption{Schematic diagram of the proof of \Cref{thm:bps_lb}. The oracle pretends to have the tournament on the left if the output of the algorithm contains any vertex from \BB; otherwise the oracle pretends to have the tournament on the right.}\label{fig:bip}
\end{figure}

Unlike \Cref{thm:many_st_lb}, we will only focus on asymptotic bounds instead of exact bounds. Hence, we can assume the number of vertices in the tournament to be even or odd without loss of generality (by adding \OO(1) vertices which lose to everyone else) in our subsequent proofs.

\begin{theorem}\label{thm:bps_lb}
 The query complexity of the bipartisan set of a tournament is $\Omega(n^2)$.
\end{theorem}

\begin{proof}
 We refer to \Cref{fig:bip} for a schematic diagram of the proof. Let $n_1$ be an odd integer. We consider a partial tournament $\TT = (\AA\cup\BB, \EE)$ where $\AA = \{a_i : i\in[n_1]\}, \BB = \{b_i : i\in[n_1]\},$ and $\EE = \{(a_i, a_{i+j \pmod {n_1}}), (b_i, b_{i+j \pmod {n_1}}): i\in[n_1], 1\le j\le \nfrac{(n_1-1)}{2}\}.$ The oracle answers the queries of the algorithm as follows. If a query comes for an edge between vertices $a_i$ and $a_j$ or $b_i$ and $b_j$ for any $i, j\in[n_1]$, then the oracle answers according to \TT. If a query comes for an edge between $a_i$ and $b_j$ for any $i,j\in[n_1]$, then the oracle says that the edge is oriented from $a_i$ to $b_j$. We now claim that the algorithm must query at least $\nfrac{n^2_1}{2}$ edges between \AA and \BB. Suppose not, then, if the output of the algorithm contains any vertex from \BB, then the oracle orients every edge between \AA and \BB, from \AA to \BB, thereby ensuring that the output of the algorithm is wrong due to \Cref{lem:bps_uc}. On the other hand, if the output of the algorithm does not contain any vertex from \BB, then the oracle orients all the edges between \AA and \BB that are not queried, from \BB to \AA, thereby ensuring that the output of the algorithm is again wrong due to \Cref{lem:heavy}. Hence, the algorithm must query at least $\nfrac{n^2_1}{2}$ edges between \AA and \BB and thus the query complexity of the bipartisan set is $\Omega(n^2).$
\end{proof}

We now extend \Cref{thm:bps_lb} to randomized algorithms.

\begin{theorem}\label{thm:random_bps_lb}
 The expected query complexity of any randomized algorithm to compute the bipartisan set of a tournament on an even number of vertices is $\Omega(n^2)$.
\end{theorem}

\begin{proof}
 Let $n_1$ be any odd integer, $\TT = (\AA\cup\BB, \EE)$ a tournament where $\AA = \{a_i : i\in[n_1]\}, \BB = \{b_i : i\in[n_1]\},$ and $\EE = \{(a_i, a_{i+j \pmod {n_1}}), (b_i, b_{i+j \pmod {n_1}}): i\in[n_1], 1\le j\le \nfrac{(n_1-1)}{2}\}\cup\{(a_i,b_j):i\in[n_1],j\in[n_1]\}.$ Let $\SS\subseteq [n_1]\times [n_1]$. Then $\TT_\SS =(\AA\cup\BB, \EE')$ is a tournament where $\EE'=\EE\setminus\{(a_i,b_j):(i,j)\in \SS\}\cup\{(b_j,a_i):(i,j)\in \SS\}$. Note that $\TT_\emptyset=\TT$. Let $\II:=\{\TT_{\SS}:\SS\subseteq [n_1]\times [n_1]\}$ be a set of input instances. Now let us define a distribution on \II. $\TT$ is considered as the input instance with probability $\frac{1}{2}$. For all $\TT'\in \II\setminus \{\TT\}$, $\TT'$ is considered as the input instance with probability $\frac{1}{2(2^{n^2_1}-1)}$. Let $\AA_1$ be a deterministic algorithm that always computes the bipartisan set.

We claim that $\AA_1$ must query at least $n_1^2/2$ edges of $\TT$. Suppose not, then there is a non-empty set $\SS'\subseteq [n_1]\times[n_1]$ such that the algorithm does not query the edges $\{(a_i,b_j):(i,j)\in \SS'\}$. Note that $|\SS'|> n^2_1/2$. Then $\AA_1$ will output the same solution for both $\TT$ and $\TT_{\SS'}$. Due to  \Cref{lem:bps_uc} and  \Cref{lem:heavy}, $\TT$ and $\TT_{\SS'}$ have different tournament solutions. Therefore, we can conclude that with non-zero probability, the algorithm $\AA_1$ will output a wrong solution. This is a contradiction to our assumption that $\AA_1$ is a deterministic algorithm that always computes the correct solution.  Hence, $\AA_1$ must query at least $n^2_1/2$ edges of $\TT$. Also recall that $\TT$ is considered as an input instance for $\AA_1$ with probability $\frac{1}{2}$. The expected query complexity is at least $\Omega(n^2)$.
\end{proof}

\begin{figure}[htbp]
\begin{center}
\begin{tikzpicture}

  \draw[ultra thick] (0,2) ellipse (.6cm and 1.5cm);
  \draw[ultra thick] (4.8,2) ellipse (.6cm and 1.5cm);
  \draw[ultra thick] (2.4,2) circle (.5cm) node (c) {};

  \node at (0,2) {$\AA$};
  \node at (4.8,2) {$\BB$};
  \node at (2.4,2) {$x$};

  \draw[<-] (2,2.4) -- (.2,3.5);
  \draw[<-] (1.9,2.3) -- (.5,3);
  \draw[<-] (1.9,2.2) -- (.7,2.5);
  \draw[<-] (1.9,2.1) -- (.7,2);
  \draw[<-] (1.9,1.9) -- (.7,1.5);
  \draw[->] (.5,1) -- (1.9,1.8);
  \draw[->] (.2,.5) -- (2,1.6);

  \draw[->] (2.8,2.4) -- (4.4,3.3);
  \draw[->] (2.9,2.3) -- (4.2,2.8);
  \draw[->] (2.9,2.2) -- (4.1,2.5);
  \draw[->] (2.9,2.1) -- (4.1,2);
  \draw[->] (2.9,1.9) -- (4.1,1.5);
  \draw[<-] (4.1,1) -- (2.9,1.8);
  \draw[<-] (4.4,.5) -- (2.9,1.6);
  \end{tikzpicture}
\end{center}
\caption{Schematic diagram of the proof of \Cref{thm:ucs_lb}. All the edges between \AA and \BB need to be queried to find whether $x$ is uncovered or not.}\label{fig:ucs}
\end{figure}

We now show that the query complexity of the uncovered set is $\Omega(n^2)$.

\begin{theorem}\label{thm:ucs_lb}
 The query complexity of the uncovered set of a tournament is $\Omega(n^2)$.
\end{theorem}

\begin{proof}
 We refer to \Cref{fig:ucs} for a schematic diagram of the proof. Consider a partial tournament $\TT = (\AA\cup\BB\cup\{x\}, \EE)$ where $n_1$ is any integer, $\AA = \{a_i : i\in[n_1]\}, \BB = \{b_i : i\in[n_1]\}$, and $\EE = \{(a_i, x), (x,b_i): i\in[n_1]\}$. Let us consider the following oracle. We define the partial tournament $\TT^\pr$ to be the graph on $\AA\cup\BB\cup\{x\}$ consisting of all the answers of the oracle till now. Hence, to begin with, $\TT^\pr$ does not have any edge. The oracle answers the queries for any edge in \TT according to \TT. For any query of the form $\{a_i, a_j\}$ or $\{b_i, b_j\}$, the oracle answers arbitrarily but consistently. For a query $\{a_i, b_j\}$ for some $i,j\in[n_1]$, if $a_i$ defeats $b_k$ for every $k\in[n_1]\setminus\{j\}$ in $\TT^\pr$, then the oracle answers that $b_j$ defeats $a_i$; otherwise the oracle answers that $a_i$ defeats $b_j$.

 We claim that any algorithm for finding the uncovered set of tournaments must query for the pair $\{a_i, b_j\}$ for every $i,j\in[n_1]$. Suppose not, then there exists a pair $\{a_i, b_j\}$ which the algorithm does not query. Notice that, by the design of the oracle, $a_i$ defeats $b_k$ in $\TT^\pr$ for every $k\in[n_1]$ such that $\{a_i, b_k\}$ has been queried by the algorithm. For every pair $\{a_t, b_\el\}$ with $t,\el\in[n_1], t\ne i$ and $\{a_t, b_\el\}$ has not been queried, the oracle orients the edge from $b_\el$ to $a_t$. The oracle also orients all the edges from $a_i$ to $b_r$ for every $r\in[n_1]\setminus\{j\}$. Now, if the output of the algorithm contains $x$, then the oracle orients the edge $\{a_i, b_j\}$ from $a_i$ to $b_j$. Notice that, in this case, $x$ is covered by $a_i$ and thus $x$ should not be in the uncovered set and hence, the output of the algorithm is wrong. On the other hand, if the output of the algorithm does not contain $x$, then the oracle orients the edge $\{a_i, b_j\}$ from $b_j$ to $a_i$. In this case, $x$ is not covered by any other vertex and thus $x$ belongs to the uncovered set. Hence, the algorithm outputs incorrectly in both the cases, thereby proving the result.
\end{proof}

\begin{figure}[htbp]
\begin{center}
\begin{tikzpicture}

  \draw[ultra thick] (0,2) ellipse (.6cm and 1.5cm);
  \draw[ultra thick] (4.8,2) ellipse (.6cm and 1.5cm);
  \draw[ultra thick] (2.4,2) circle (.5cm) node (c) {};

  \node at (0,2) {$\AA$};
  \node at (5.8,2) {$\BB$};
  \node at (4.8,3) {$b_1$};
  \node at (4.8,2.5) {$b_2$};
  \node at (4.8,2.1) {$\cdot$};
  \node at (4.8,1.7) {$\cdot$};
  \node at (4.8,1.2) {$\cdot$};
  \node at (4.8,.8) {$b_n$};
  \node at (2.4,2) {$x$};

  \draw[<-] (2,2.4) -- (.2,3.5);
  \draw[<-] (1.9,2.3) -- (.5,3);
  \draw[<-] (1.9,2.2) -- (.7,2.5);
  \draw[<-] (1.9,2.1) -- (.7,2);
  \draw[<-] (1.9,1.9) -- (.7,1.5);
  \draw[->] (.5,1) -- (1.9,1.8);
  \draw[->] (.2,.5) -- (2,1.6);

  \draw[->] (2.8,2.4) -- (4.4,3.3);
  \draw[->] (2.9,2.3) -- (4.2,2.8);
  \draw[->] (2.9,2.2) -- (4.1,2.5);
  \draw[->] (2.9,2.1) -- (4.1,2);
  \draw[->] (2.9,1.9) -- (4.1,1.5);
  \draw[<-] (4.1,1) -- (2.9,1.8);
  \draw[<-] (4.4,.5) -- (2.9,1.6);
  \end{tikzpicture}
\end{center}
\caption{Schematic diagram of the proof of \Cref{thm:banks_lb}. $b_1, b_2, \ldots, b_n$ is a transitive order of the tournament induced on \BB. All the edges between \AA and \BB need to be queried to find whether the Banks set of the input tournament contains $x$ or not.}\label{fig:banks}
\end{figure}

Next we consider the Banks set and show its query complexity to be $\Omega(n^2).$

\begin{theorem}\label{thm:banks_lb}
 The query complexity of the Banks set of a tournament is $\Omega(n^2)$.
\end{theorem}

\begin{proof}
 We refer to \Cref{fig:banks} for a schematic diagram of the proof. Consider a partial tournament $\TT = (\AA\cup\BB\cup\{x\}, \EE),$ where $n_1$ is any integer, $\AA = \{a_i : i\in[n_1]\}, \BB = \{b_i : i\in[n_1]\}$, and $\EE = \{(a_i, x), (x,b_i):i\in[n_1]\}\cup\{(b_i, b_j): i,j \in[n_1], i<j\}.$ Let us consider the following oracle. Let us define the partial tournament $\TT^\pr$ to be the graph on $\AA\cup\BB\cup\{x\}$ consisting of all the answers of the oracle till now. Hence, to begin with, $\TT^\pr$ does not have any edge. The oracle answers the queries for any edge in \TT according to \TT. For any query of the form $\{a_i, a_j\}$, the oracle answers arbitrarily but consistently. For a query $\{a_i, b_j\}$ for some $i,j\in[n_1]$, if $a_i$ defeats $b_k$ for every $k\in[n_1]\setminus\{j\}$ in $\TT^\pr$, then the oracle answers that $b_j$ defeats $a_i$; otherwise the oracle answers that $a_i$ defeats $b_j$.

 We claim that the algorithm must query for the pair $\{a_i, b_j\}$ for every $i,j\in[n_1]$. Suppose not, then there exists a pair $\{a_i, b_j\}$ which the algorithm does not query. Notice that, by the design of the oracle, $a_i$ defeats $b_k$ in $\TT^\pr$ for every $k\in[n_1]$ such that $\{a_i, b_k\}$ has been queried by the algorithm. For every pair $\{a_t, b_\el\}$ with $t,\el\in[n_1], t\ne i$ and $\{a_t, b_\el\}$ has not been queried, the oracle orients the edge from $b_\el$ to $a_t$. The oracle also orients all the edges not in $\TT^\pr$ between $a_i$ and $b_r$ from $a_i$ to $b_r$ for every $r\in[n_1]\setminus\{j\}$.

 Now if the output of the algorithm contains $x$, then the oracle orients the edge between $a_i$ and $b_j$ from $a_i$ to $b_j$. We claim that $x$ can not be the maximum vertex of any maximal transitive sub-tournament $\TT^\prr$ of \TT. To see this, we first observe that the sub-tournament $\TT^\prr$ must have all the vertices in \BB and no vertex from \AA. Indeed, otherwise either $x$ is not the maximum vertex of $\TT^\prr$ (if any vertex from \AA is there in $\TT^\prr$) or $\TT^\prr$ is not a maximal transitive sub-tournament (if any vertex from \BB is not there in $\TT^\prr$). However, such a sub-tournament is not a maximal transitive sub-tournament since $a_i$ can be added to $\TT^\prr$ without violating transitivity. Hence, $x$ does not belong to the Banks set of the resulting tournament and thus the algorithm's output is wrong.

 On the other hand, suppose the output of the algorithm does not contain $x$. Then the oracle orients the edge between $a_i$ and $b_j$ from $b_j$ to $a_i$. In this case, the sub-tournament of \TT induced on $\BB\cup\{x\}$ is a maximal sub-tournament where $x$ is the maximum vertex and thus $x$ belongs to the Banks set of the resulting tournament. Hence, the algorithm outputs incorrectly in both the cases, thereby proving the result.
\end{proof}

We now extend \Cref{thm:ucs_lb,thm:banks_lb} to randomized algorithms.

\begin{theorem}\label{thm:random_ucs_bs_lb}
The expected query complexity of any randomized algorithm to compute the uncovered set and Banks set of a tournament is $\Omega(n^2)$.
\end{theorem}

\begin{proof}
 The argument for both the uncovered set and Banks set are the same. So we present them together. Let $\AA_1$ be a deterministic algorithm that always computes the uncovered set (or Banks set). Let $\TT=(\AA\cup\BB\cup\{x\}, \EE)$ be a partial tournament where $\AA = \{a_i : i\in[n_1]\}, \BB = \{b_i : i\in[n_1]\}$ and $\EE = \{(a_i, x), (x,b_i), (a_i,a_j),(b_i, b_j): i,j \in[n_1], i>j\}$. Let $\II_0$ be the set of all tournaments that contain the partial tournament $\TT$, and for all $i\in [n_1]$, there is exactly one $j\in [n_1]$ such that there is an edge from $b_j$ to $a_i$. For all $k\in[n_1]$, let $\II_k$ be the set of all tournaments that
 \begin{itemize}
 	\item contain the partial tournament $\TT$,
 	\item for all $i\in [n_1]\setminus \{k\}$, there is exactly one $j\in [n_1]$ such that there is an edge from $b_j$ to $a_i$, and
 	\item for all $j\in[n_1]$, there is an edge from $a_k$ to $b_j$.
 \end{itemize}
 Clearly $|\II_0|=n_1^{n_1}$, and for all $k\in [n_1]$, $|\II_k|=n_1^{n_1-1}$ Let $\II:=\bigcup_{i=0}^{n}\II_i$ be the set of input instances for $\AA_1$; $|\II|=\sum_{k=0}^n|\II_k|=2n_1^{n_1}$. Now let us define a distribution on \II. For all $\TT'\in\II$, $\TT'$ is considered as an input instance for $\AA_1$ with probability $\frac{1}{2n_1^{n_1}}$.

 Let $\II'$ denote the input instance and $\QQ$ be the number of queries made by $\AA_1$ on the input. Now we show that $\mathbb{E}[\QQ|\II'\in \II_0]$ is $\Omega(n^2)$. If $\II'\in \II_0$, then the algorithm has to query all the edges from $\BB$ to $\AA$. Suppose an edge $(b_j,a_i)$ is not queried. Then there is an instance $\II''\in \II\setminus\{\II'\}$ which we get by reversing the edge $(b_j,a_i)$ in $\II'$. $\AA_1$ outputs the same solution for both $\II'$ and $\II''$. As $\II'$ and $\II''$ have different tournament solutions (refer to the proof of \Cref{thm:ucs_lb,thm:banks_lb}), we can conclude that with non-zero probability, the algorithm $\AA_1$ will output a wrong solution. This is a contradiction to our assumption that $\AA_1$ is a deterministic algorithm that always computes the correct solution.  Hence, $\AA_1$ must query all the edges from $\BB$ to $\AA$ in the worst case.

 We next show that the expected query complexity of $\AA_1$ is $\Omega(n^2)$ when the input is chosen uniformly randomly from \II. Towards that, we introduce a few notation to carefully compute $\mathbb{E}[\QQ|\II'\in \II_0]$. Let $\II'\in \II_0$ be the input instance (chosen randomly) and $X_i$ a random variable which denotes the number of edges queried from $a_i$ to \BB by $\AA_1$ until it queries an edge from \BB to $a_i$. Let $Y_i^j$ be an indicator variable which is $1$ when the $j$-th edge to be queried between $a_i$ and \BB is an edge from \BB to $a_i$. Let $\HH_i^j$ be the set of all possible sequences $e_1,\ldots,e_{j-1}$ of $(j-1)$ edges between $a_i$ and $\BB$ which the algorithm $\AA_1$ queries for the instances in $\II_0$ before querying an edge between $a_i$ and \BB for the $j$-th time such that none of the edges in $\{e_1,\ldots,e_{j-1}\}$ is from \BB to $a_i$. Let $\HH\in \HH_i^j$ be any such sequence. After querying the sequence of edges $\HH$, let $\AA_1$ query an edge between $a_i$ and $b_{k_{\HH}}$ where $k_{\HH}$ depends only on $\HH$ as $\AA_1$ is deterministic. Now observe that the total number of instances in $\II_0$ which contain the edges $\HH$ is $(n_1-j+1)$ times the total number of instances in $\II_0$ which contain the edges $\HH\cup\{(b_{k_{\HH}},a_i)\}$ as the edge from $\BB$ to $a_i$ could be any one of the $(n_1-j+1)$ edges which were not queried by $\AA_1$. Therefore $\Pr[(b_{k_{\HH}},a_i)\in\EE[\II']|X_i\geq j-1, \II'\in \II_0,\HH]=\frac{1}{n_1-j+1}$. Hence, we have the following:
 \begin{align*}
 	&\Pr[X_i=j-1|X_i\geq j-1, \II'\in\II_0]\\
 	&=\sum_{\HH\in \HH_i^j}\Pr[\HH|X_i\geq j-1, \II'\in\II_0]\cdot\Pr[(b_{k_{\HH}},a_i)\in\EE[\II']|X_i\geq j-1, ,\II'\in \II_0,\HH]\\
 	&=\frac{1}{n_1-j+1}\sum_{\HH\in \HH_i^j}\Pr[\HH|X_i\geq j-1, \II'\in\II_0]\\
 	&=\frac{1}{n_1-j+1}
 \end{align*}
 Fix a number $j'\in [n_1-1]$. Now we show that $\Pr[X_i=j'|\II'\in\II_0]=\frac{1}{n_1}$ if $\Pr[X_i=k-1|\II'\in\II_0]=\frac{1}{n_1}$, $\forall k\in[j']$.
 \begin{align*}
 	\Pr[X_i=j'|\II'\in\II_0]&=\Pr[X_i=j'|X_i\ge j^\pr,\II'\in\II_0]\cdot \Pr[X_i\ge j'|\II'\in\II_0]\\
 	&=\frac{1}{n_1-j'}\cdot\left(1-\Pr[X_i<j'|\II'\in\II_0]\right)\\
 	&=\frac{1}{n_1-j'}\cdot \left(1-\sum_{k\in[j']}\Pr[X_i=k-1|\II'\in\II_0]\right)\\
 	&=\frac{1}{n_1-j'}\cdot \left(1-\frac{j'}{n_1}\right)\\
 	&=\frac{1}{n_1-j'}\cdot \frac{n_1-j'}{n_1}\\
 	&=\frac{1}{n_1}
 \end{align*}
 Now observe that $\Pr[X_i=0| \II'\in\II_0]=\Pr[X_i=0|X_i\geq 0, \II'\in\II_0]=\frac{1}{n_1}$. Hence, $\Pr[X_i=j-1|\II'\in\II_0]=\frac{1}{n_1}$ for all $j\in[n_1]$.

 Hence, $\mathbb{E}[X_i|\II'\in\II_0]=\sum_{j=0}^{n_1-1}\frac{j}{n_1}=\frac{n_1-1}{2}$. Now observe that $\mathbb{E}[\QQ|\II'\in\II_0]\geq \sum_{i=1}^{n_1}\mathbb{E}[X_i|\II'\in\II_0]=\frac{n_1(n_1-1)}{2}$. Hence, we have $\mathbb{E}[\QQ]=\frac{1}{2}\cdot \mathbb{E}[\QQ|\II'\in\II_0]=\frac{n_1(n_1-1)}{2}$. Therefore, the expected query complexity is at least $\Omega(n^2)$.
\end{proof}

\begin{figure}[htbp]
\begin{center}
\begin{tikzpicture}

  \draw[ultra thick] (0,2) ellipse (.6cm and 1.5cm);
  \draw[ultra thick] (4.8,2) ellipse (.6cm and 1.5cm);

  \node at (-1,2) {$\AA$};
  \node at (5.8,2) {$\BB$};

  \node at (4.8,3) {$b_1$};
  \node at (4.8,2.5) {$b_2$};
  \node at (4.8,2.1) {$\cdot$};
  \node at (4.8,1.7) {$\cdot$};
  \node at (4.8,1.2) {$\cdot$};
  \node at (4.8,.8) {$b_n$};

  \node at (0,3) {$a_1$};
  \node at (0,2.5) {$a_2$};
  \node at (0,2.1) {$\cdot$};
  \node at (0,1.7) {$\cdot$};
  \node at (0,1.2) {$\cdot$};
  \node at (0,.8) {$a_n$};
  \end{tikzpicture}
\end{center}
\caption{Schematic diagram of the proof of \Cref{thm:topcycle_lb}. There are two cycles namely (i) $a_1 \rightarrow a_2 \rightarrow \cdots \rightarrow a_i \rightarrow a_{i+1} \rightarrow \cdots \rightarrow a_n \rightarrow a_1$ and (ii) $b_1 \rightarrow b_2 \rightarrow \cdots \rightarrow b_i \rightarrow b_{i+1} \rightarrow \cdots \rightarrow b_n \rightarrow b_1$. All the edges between \AA and \BB need to be queried to determine the top cycle of the input tournament.}\label{fig:tc}
\end{figure}

We now show that the query complexity of the top cycle of tournaments is $\Omega(n^2).$

\begin{theorem}\label{thm:topcycle_lb}
 The query complexity of the top cycle of a tournament is $\Omega(n^2)$.
\end{theorem}

\begin{proof}
 We consider a partial tournament $\TT = (\AA\cup\BB, \EE)$ where $\AA = \{a_i : i\in[n_1]\}, \BB = \{b_i : i\in[n_1]\},$ and $\EE = \{(a_i, a_{i+1 \pmod {n_1}}), (b_i, b_{i+1 \pmod {n_1}}): i\in[n_1]\}.$ The oracle answers the queries of the algorithm as follows. If a query comes for the edge between vertices $a_i$ and $a_j$ or $b_i$ and $b_j$ for any $i, j\in[n_1]$, then the oracle answers according to \TT if the edge is present in \TT, and arbitrarily but consistently otherwise. If a query comes for an edge between $a_i$ and $b_j$ for any $i,j\in[n_1]$, then the oracle says that the edge is oriented from $a_i$ to $b_j$. Now we claim that the algorithm must query all the $n_1^2$ edges between \AA and \BB. Suppose not, then there exist $a_i$ and $b_j$ for some $i, j\in[n_1]$ such that the algorithm has not queried for the edge between $a_i$ and $b_j$. Now if the output of the algorithm does not contain any vertex from \BB, then the oracle orients the edge between $a_i$ and $b_j$ from $b_j$ to $a_i$. Notice that, in this case the top cycle of the resulting tournament \TT contains at least one vertex $b_j\in\BB$ and thus the algorithm does not output correctly in this case. On the other hand, if the output of the algorithm contains any vertex from $\BB$, then the oracle orients all the edges between \AA and \BB from \AA to \BB. In this case, the top cycle of the resulting tournament is \AA and thus the algorithm again fails to output correctly. Hence, the algorithm must make $\Omega(n^2)$ queries.
\end{proof}

We now extend \Cref{thm:topcycle_lb} to randomized algorithms.

\begin{theorem}\label{thm:random_topcycle_lb}
 The expected query complexity of any randomized algorithm to compute the top cycle of a tournament is $\Omega(n^2)$.
\end{theorem}
\begin{proof}
 We consider a tournament $\TT_1$ which contains the partial tournament $\TT = (\AA\cup\BB, \EE)$ where $n_1$ is any integer, $\AA = \{a_i : i\in[n_1]\}, \BB = \{b_i : i\in[n_1]\},$ and $\EE = \{(a_i, a_{i+1 \pmod {n_1}}), (b_i, b_{i+1 \pmod {n_1}}): i\in[n_1]\}\cup\{(a_i,b_j):i\in[n_1],j\in[n_1]\}.$ Let $\TT_{1}^{u,v}$ be the tournament which is the same as $\TT_{1}$ except the edge from $u$ to $v$ is reversed, that is, $\TT_{1}^{u,v} = \TT_{1}\cup\{(v,u)\}\setminus\{(u,v)\}$. Let $\II:=\{\TT_{1}\}\cup\{\TT_{1}^{a_i,b_j}:i\in[n_1],j\in[n_1]\}$ be the set of input instances for $\AA_1$. Now let us define a distribution on \II. $\TT_{1}$ is considered as an input instance for $\AA_1$ with probability $\frac{1}{2}$. For all $\TT'\in \II\setminus \{\TT_{1}\}$, $\TT'$ is considered as an input instance for $\AA_1$ with probability $\frac{1}{2n_1^2}$. Let $\AA_1$ be a deterministic algorithm that always computes the top cycle.

We claim that $\AA_1$ must query at least $n_1^2$ edges of $\TT_{1}$. Suppose not, then there exists an edge from $a_i$ to $b_j$ in $\TT_{1}$ for some $i,j \in [n_1]$ that the algorithm does not query. Then $\AA_1$ will output the same solution for both $\TT_{1}$ and $\TT_{1}^{a_i,b_j}$. As $\TT_1$ and $\TT_{1}^{a_i,b_j}$ have different top cycles, we can conclude that with non-zero probability, the algorithm $\AA_1$ will output a wrong solution. This is a contradiction to our assumption that $\AA_1$ is a deterministic algorithm that always computes the correct solution.  Hence, $\AA_1$ must query at least $n^2$ edges of $\TT_{1}$. Also recall that $\TT_{1}$ is considered as an input instance for $\AA_1$ with probability $\frac{1}{2}$. The expected query complexity is at least $\Omega(n^2)$.
\end{proof}

\section{Results for Tournaments with Small Top Cycle}\label{sec:small_cycle}

It turns out that, if we a priori know that the size of the top cycle in the input tournament is at most $k$, then there is an algorithm for finding the top cycle with a smaller number of queries.

\begin{theorem}\label{thm:topcycle_ub}
 Suppose we know that the top cycle of the input tournament $\TT$ is of size at most $k$. Then there exists an algorithm for finding the top cycle of $\TT$ with query complexity $\BigO(nk + \nfrac{n\log (\nfrac{n}{k}) }{\log(\nfrac{k}{k-1})})$.
\end{theorem}

\begin{proof}
 We first partition the set of vertices \VV into $\lceil \nfrac{n}{2k} \rceil$ subsets $\VV_i, i\in[\lceil \nfrac{n}{2k} \rceil]$ such that for every $i\in[\lfloor \nfrac{n}{2k} \rfloor]$, $|\VV_i|=2k$. For each part of the partition, we query all pairs of vertices. We notice that for every $i\in[\lfloor \nfrac{n}{2k} \rfloor]$, there must exist at least one vertex $v_i\in \VV_i$ with in-degree at least $k$. Hence, $v_i$ does not belong to the top cycle of \TT since the in-degree of every vertex in the top cycle is at most size of the top cycle minus $1$ which is $k-1$ here. We delete the vertex $v_i$ from $\VV_i$ for every $i\in [\lfloor \nfrac{n}{2k} \rfloor]$, thereby deleting $\lfloor\nfrac{n}{2k}\rfloor$ vertices in total. We now iterate the same process on the remaining set of vertices. The first iteration takes $\OO((\nfrac{n}{k}) k^2) = \OO(nk)$ queries. From the next iteration onwards, we can manage with at most $n$ queries per iteration by partitioning the vertices cleverly: since we have deleted exactly one vertex from each set of the partition we can add one vertex to each set of the partition by ``breaking'' some of the sets from the partition. Formally we do the following. Let $\VV_1,\VV_2,\ldots,\VV_\ell$ be the partition at the end of the previous iteration. Note that for every $i\in[\ell-1]$, $|\VV_i|=2k-1$ and $|\VV_\ell|\leq 2k-1$. We now add one vertex each to $\VV_1,\VV_2,\ldots$ from $\VV_\el$. If $\VV_\el$ becomes an empty set in the process, then we continue using the vertices from $\VV_{\el-1}$ and so on. At the end of the process, we have a partition $\VV_1^\pr,\VV_2^\pr,\ldots,\VV^\pr_{\ell^\pr}$ such that $|\VV_i^\pr\setminus\VV_i|=1$ for every $i\in[\el^\pr-1]$ and we have $|\VV_{\ell^\pr}^\pr\setminus\VV_{\ell^\pr}|=1$ or $\VV_{\ell^\pr}^\pr\subseteq\VV_{\ell^\pr}$. We now observe that, in each of the sets $\VV_1^\pr,\ldots \VV_{\ell^\pr-1}^\pr$, we now need to compare the newly added vertex with the rest of the vertices in order to find the vertex with in-degree at least $k$. For $\VV_{\ell^\pr}^\pr$, if $|\VV_{\ell^\pr}^\pr\setminus\VV_{\ell^\pr}|=1$, then we now need to compare the newly added vertex with the rest of the vertices in order to find the vertex with in-degree at least $k$ in $\VV_{\ell^\pr}^\pr$; or otherwise (if $\VV_{\ell^\pr}^\pr\subseteq\VV_{\ell^\pr}$) we already made all pairwise comparisons in $\VV_{\ell^\pr}^\pr$ in the last iteration. Therefore, we require at most $n$ queries in total. Since, each iteration decreases the size of the tournament by a factor of $\frac{n}{n-\lfloor\frac{n}{2k}\rfloor}=\Omega(\nfrac{k}{k-1}),$ after $\OO(\nfrac{\log (\nfrac{n}{k}) }{\log(\nfrac{k}{k-1})})$ iterations, we have $\OO(k)$ vertices in the tournament where we can find the top cycle using $\OO(k^2) = \OO(nk)$ queries. Hence, the query complexity of our algorithm is $\OO(nk + \nfrac{n\log (\nfrac{n}{k}) }{\log(\nfrac{k}{k-1})}).$ The correctness of the algorithm follows from the fact that whenever we remove a vertex $v$ from the tournament, $v$ does not belong to the top cycle of \TT.
\end{proof}

The following result gives the relationships between the top cycle of a tournament and other tournament solutions like the Copeland set, the Slater set, the Markov set, the bipartisan set, the uncovered set, and the Banks set.

\begin{lemma}\label{lem:others_cycle}
 Let \TT be a tournament whose top cycle is \CC. Then the Copeland set, the Slater set, the Markov set, the bipartisan set, the uncovered set, and the Banks set of $\TT$ are the same as the corresponding solutions for the tournament $\TT(\CC)$ which is the tournament \TT restricted to \CC.
\end{lemma}

\begin{proof}
\begin{itemize}
 \item {\bf Copeland set:} The Copeland set is a subset of $\TT(\CC)$ since every vertex in \CC covers every vertex in $\TT\setminus\CC$ and thus the out-degree of every vertex in $\TT(\CC)$ is strictly more than the out-degree of every vertex in $\TT\setminus\CC$. Also, for any vertex $v$ in $\TT(\CC)$, the out-degree of $v$ in $\TT(\CC)$ is its out-degree in \TT minus the number of vertices in $\TT\setminus\CC$. Hence, a vertex $v$ in $\TT(\CC)$ has the maximum out-degree in $\TT(\CC)$ if and only if $v$ has the maximum out-degree in \TT.

 \item {\bf Uncovered set:} The uncovered set of \TT is a subset of $\TT(\CC)$ since every vertex in \CC covers every vertex in $\TT\setminus\CC$ and thus, no vertex in $\TT\setminus\CC$ belongs to the uncovered set of \TT. Also, for every vertex $v$ in $\CC$, the dominion $D_{\CC}(v)$ of $v$ in $\TT(\CC)$ is its dominion $D_\TT(v)$ in \TT set minus $(\TT\setminus\CC)$ and $(\TT\setminus\CC)\subseteq D_\TT(v)$. Hence, for every vertex $v\in\CC$, $v$ is uncovered in \TT if and only if it is uncovered in $\TT(\CC)$.

 \item {\bf Bipartisan set} The bipartisan set is a subset of $\TT(\CC)$ thanks to \Cref{lem:bps_uc}. Let $p_\TT\in\Delta(\TT)$ be the unique maximal lottery of \TT. Thus we have $BP(\TT)=\{a\in\TT: p_\TT(a)>0\}$ and
 \[\sum_{a\in\TT} p_\TT(a)g_{ab}\ge 0 \text{ for all }b\in\TT\]
 where $(g_{ab})_{a,b\in\TT}$ is the skew-symmetric matrix of \TT. We define $q\in\Delta(\TT(\CC))$ to be $q(a)=p_\TT(a)$. Hence, we have
 \[\sum_{a\in\TT(\CC)} q(a)g_{ab}\ge 0 \text{ for all }b\in\TT.\]
 Hence, $q$ is the maximal lottery of $\TT(\CC)$. We now have $BP(\TT(\CC))=\{a\in\TT(\CC): q(a)>0\}=\{a\in\TT: p_\TT(a)>0\}=BP(\TT)$ which proves the lemma for bipartisan set.

 \item {\bf Markov set:} All the states corresponding to the vertices in $\VV\setminus\CC$ are inessential and thus do not belong to the Markov set of \TT. Let $\pi_\TT$ be the unique stationary distribution of the Markov chain induced by \TT. Then we have $\pi_\TT(a)=0$ for every $a\in\TT\setminus\TT(\CC)$. We now define a distribution $q\in\Delta(\TT(\CC))$ as $q(a)=\pi_\TT(a)$ for every $a\in\TT(\CC)$. Clearly $q$ is the unique stationary distribution of the Markov chain induced by $\TT(\CC)$. We now have $MA(\TT(\CC))=\{a\in\TT(\CC): q(a)>0\}=\{a\in\TT: p_\TT(a)>0\}=MA(\TT)$ which proves the lemma for Markov set.

 \item {\bf Slater set:} We observe that, in the Slater order \suc of the tournament \TT, every vertex in \CC must be preferred over every vertex in $\VV\setminus\CC$. If not, then let there be a vertex $a\in\CC$ and $b\in\VV\setminus\CC$ such that $a$ immediately follows $b$ in \suc. Then by swapping the positions of the vertices $a$ and $b$ in \suc, we can strictly decrease the disagreement of \suc with \TT, thereby contradicting that \suc is a Slater order of \TT. Finally we observe that, since all the vertices in $\TT(\CC)$ appear before every other vertex in every Slater order, $\suc$ is a Slater order of \TT if and only if $\suc\divides_{\TT(\CC)}$ ($\suc$ restricted to $\TT(\CC)$) is a Slater of $\TT(\CC)$. This proves the lemma for the Slater set.

 \item {\bf Banks set:} Let $v$ be an element in the Banks set of \TT. Hence, there exists a maximal transitive sub-tournament $\TT^\pr$ of \TT where $v$ is the maximal element. We claim that $v\in\CC$. Suppose not, then since every element in \CC dominates $v$ and every other elements in $\TT^\pr$ as $v$ dominates them and $v$ dominates no element in \CC. However, this contradicts our assumption that $\TT^\pr$ is a maximal transitive sub-tournamnet of \TT. Hence, we have $v\in\TT^\pr\cap \CC$. We now observe that $\TT^\pr\cap\CC$ is also a maximal transitive sub-tournament of \CC with $v$ as its maximal element. Suppose not, then there exists an element $v^\pr$ in \CC which dominates every element in $\TT^\pr\cap\CC$. Now since every element in \CC dominates every element in $\TT\setminus\CC$, $v^\pr$ dominates every element in $\TT^\pr$ contradicting its maximality. Hence, $v$ belongs to the Banks set of \CC.

 Now let us take any element $w$ from the Banks set of \CC. Hence, there exists a maximal transitive sub-tournament $\TT^\prr$ of \CC where $w$ is the maximal element. We keep adding elements from $\TT\setminus\CC$ in $\TT^\prr$ to obtain another maximal transitive sub-tournament $\TT_1$ of \TT. We claim that the maximal element of $\TT_1$ is $w$. Towards that, we first observe that every element in $\TT_1\setminus\TT^\prr$ belongs to $\TT\setminus\CC$ since $\TT^\prr$ is a maximal transitive sub-tournament of \CC. However, $w$ dominates every element in $\TT\setminus\CC$ since $w$ belongs to the top cycle. Hence, the maximal element of $\TT_1$ is $w$ and thus $w$ belongs to the Banks set of \TT.$\hfill\qedhere$
\end{itemize}
\end{proof}

\Cref{lem:others_cycle,thm:topcycle_ub} immediately give the following query complexity upper bound for the Copeland set, the Slater set, the Markov set, the bipartisan set, the uncovered set, and the Banks set when we a priori know that the size of the top cycle of the input tournament is at most $k$.

\begin{theorem}\label{thm:others_cycle_ub}
 Suppose we know that the top cycle of the input tournament $\TT$ is of size at most $k$. Then there exists an algorithm for finding the Copeland set, the Slater set, the Markov set, the bipartisan set, the uncovered set, and the Banks set of $\TT$ with query complexity $\BigO(nk + \nfrac{n\log (\nfrac{n}{k}) }{\log(\nfrac{k}{k-1})})$.
\end{theorem}

\begin{proof}
 We first find the top cycle \CC of \TT using \Cref{thm:topcycle_ub}. This step requires $\BigO(nk + \nfrac{n\log (\nfrac{n}{k}) }{\log(\nfrac{k}{k-1})})$ queries. Next, we query for all the pair of vertices in \CC and output the corresponding solution of $\TT(\CC).$ The correctness of the algorithm follows immediately from \Cref{lem:others_cycle}. Since the second step requires $\BigO(k^2) = \OO(nk)$ queries, the query complexity of our algorithm is $\BigO(nk + \nfrac{n\log (\nfrac{n}{k}) }{\log(\nfrac{k}{k-1})})$.
\end{proof}

\section{Conclusion and Future Directions}\label{sec:con}

We have shown that, for finding many common tournament solutions, one has to query, in the worst case, almost the entire set of edges in the tournament. On the positive side, we have exhibited an important structural property, in terms of the top cycle of the tournament being small, which substantially reduces the query complexity of common tournament solutions.
An immediate future direction of research is to study the query complexity for other tournament solutions like the minimal covering set, the minimal extending set, the minimal TC-retentive set, the tournament equilibrium set, etc. Finding other structural properties of the tournament that can be leveraged to reduce the query complexity of common tournament solutions is another important direction of future research. Another interesting future direction of research is to study the query complexity of the problem of finding any vertex in a tournament solution; Goyal et al.~\cite{goyalCaldem17} already showed that, for the Copeland rule, the query complexity of finding any vertex in the Copeland set remains $\Omega(n^2)$. However, our other proofs heavily rely on the fact that the algorithm needs to output the exact tournament solution.

\longversion{\section*{Acknowledgement:} Palash Dey gratefully acknowledges all the useful discussions with Neeldhara Misra during the course of the work. Palash Dey also thanks Google Pvt. Ltd. for providing him a travel grant to enable him to attend AAAI 2017 and present the paper.}

\bibliographystyle{alpha}
\bibliography{query}

\end{document}